\documentclass[proceedings]{stacs}
\stacsheading{2010}{179-190}{Nancy, France}
\firstpageno{179}

\usepackage{cite,amsthm,amsmath,amssymb,latexsym}

\theoremstyle{definition}\newtheorem{fact}{Fact}

\newcommand{\ceil}[1]{\left\lceil #1 \right\rceil}

\newcommand{\ri}{\mbox{\it Recent-items}}
\newcommand{\ei}{\mbox{\it Expire}}
\newcommand{\eif}{\mbox{\it F-Expire}}

\newcommand{\ns}{c}
\newcommand{\hns}{\hat c}

\newcommand{\cs}{c}
\newcommand{\hcs}{\hat c}
\newcommand{\kcs}{\tilde c}

\newcommand{\rs}{r_\sigma}
\newcommand{\cjs}{c_{j}}

\newcommand{\hcjs}{\hat c_{j}}
\newcommand{\rjs}{r_{j,\sigma}}

\newcommand{\nsm}{c_\sigma}

\newcommand{\cjsm}{c_{j,\sigma}}

\newcommand{\off}{\mbox{\it off}_{j}}
\newcommand{\true}{\mbox{\it true}}
\newcommand{\false}{\mbox{\it false}}
\newcommand{\comment}[1]{}

\begin{document}

\title[Continuous Monitoring of Distributed Data Streams over Sliding Window]
{Continuous Monitoring of Distributed Data Streams over a Time-based Sliding Window}

\author[lab1]{H.L. Chan}{Ho-Leung Chan}
\address[lab1]{Department of Computer Science, University of Hong Kong, Hong Kong}  
\email{{hlchan, twlam, hfting}@cs.hku.hk}
\author[lab1]{T.W. Lam}{Tak-Wah Lam}
\author[lab2]{L.K. Lee}{Lap-Kei Lee}
\address[lab2]{Max-Planck-Institut f\"ur Informatik, 66123 Saarbr\"ucken, Germany}	
\email{lklee@mpi-inf.mpg.de}
\author[lab1]{H.F. Ting}{Hing-Fung Ting}
\thanks{T.W. Lam is partially supported by the GRF Grant HKU-713909E;
H.F. Ting is partially supported by the GRF Grant HKU-716307E}



\begin{abstract}
\noindent
The past decade has witnessed many interesting algorithms for maintaining
statistics over a data stream.
This paper initiates a theoretical study of algorithms for
monitoring distributed data streams over a time-based sliding window
(which contains a variable number of items
and possibly out-of-order items).  The concern is how to
minimize the communication between individual streams and the root,
while allowing the root, at any time, to be able to report the
global statistics of all streams within a given error bound.  This
paper presents communication-efficient algorithms for three
classical statistics, namely, basic counting,
frequent items and quantiles.
The worst-case communication cost over a window is
$O(\frac{k}{\varepsilon} \log \frac{\varepsilon N}{k})$
bits for basic counting and
$O(\frac{k}{\varepsilon} \log \frac{N}{k})$ words for the remainings, where
$k$ is the number of distributed data streams, $N$ is
the  {total}
 number of items in the streams that arrive or expire in the
window, and
$\varepsilon < 1$ is the desired error bound. Matching and
nearly matching lower bounds are also obtained.
\end{abstract}

\maketitle

\section{Introduction}
The problems studied in this paper are best illustrated by the following
puzzle. John and Mary work in different laboratories and communicate by
telephone only.  In a forever-running experiment, John records which
devices have an exceptional signal in every 10 seconds.
To adjust her devices, Mary at any time needs to keep track of the number of
exceptional signals generated by each device of John in the last one hour.
John can call Mary every 10 seconds to report the exceptional signals,
yet this requires too many calls in an hour and the total
message size per hour is linear to
the total number $N$ of exceptional signals in an hour.
Mary's devices actually allow some small error.  Can the number of
calls and message size be
reduced to $o(N)$, or even poly-$\log N$ if a small error (say, 0.1\%)
is allowed?
It is important to note that the input is given online
and Mary needs to know the answers continuously;
this makes our problem different from those in other similar classical
models, such as the
Simultaneous Communication
Complexity model \cite{BabaiGKL04}, in which all inputs
are given in advance and the parties need to
compute an answer only once.

\vspace{0.5ex}
{\bf Motivation.}
The above problem appears in data stream applications, e.g.,
network monitoring or stock analysis.
In the last decade, algorithms for continuous monitoring of a
single  massive data stream gained
a lot of attention (see \cite{Aggarwal06,Muthukrishnan05} for a survey), and
the main challenge has been how to represent the massive data
using limited space, while allowing
certain statistics (e.g., item counts, quantiles)
to be computed with sufficient accuracy.

The space-accuracy tradeoff for representing a single stream
has gradually been
understood over the years (e.g., \cite{AlonMS02,Indyk00,GuhaKS01,DemaineLM02}).
Recently, motivated by large scale networks,
the database community is enthusiastic about
communication-efficient algorithms for
continuous monitoring of
multiple, distributed data streams.
In such applications, we have $k \ge 1$ remote sites
each monitoring a data stream, and there is a
root (or coordinator) responsible for computing some global
statistics.  A remote site needs to
maintain certain statistics itself, and has to
communicate with the root often enough so that
the root can compute, at any time,
the statistics of the union of all
data streams within a certain error.
The objective is to minimize the communication.
The communication aspects of data streams introduce several challenging
theoretical questions such as
what is the optimal communication-accuracy tradeoff for maintaining a
particular statistic, and whether two-way communication is inherently
more efficient than one-way communication.

\vspace{0.5ex}
{\bf Data stream models and {\boldmath $\varepsilon$}-approximate queries.}
The data stream at
each remote site is a sequence of items from a totally ordered set $U$.
Each item is associated with
an integral time-stamp recording its  {arrival} time.
Each remote site has limited space and hence it
can only maintain the required statistics approximately.
The statistics can be
based on the whole data stream
\cite{AlonMS02,Indyk00,GuhaKS01,DemaineLM02} or
only the recent items
\cite{DatarM02,ArasuM04,LeeT06b}.
Recent items can be modeled by two types of sliding windows
\cite{BabcockDM02,DatarGIM02}.
Let $W$ be the window size, which is a positive integer.
The \emph{count-based sliding window}
includes the last $W$ items in the data stream, while the
\emph{time-based sliding window}
includes items whose time-stamps are within the last $W$ time units.
The latter assumes that zero or more items can arrive at a time.
Items in a sliding window will expire and are more
difficult to handle than in the whole data stream.
For example, counting the frequency of a
certain item in the whole stream
can be done easily by maintaining a single counter, yet the same
problem requires space
$\Theta(\frac{1}{\varepsilon} \log^2 (\varepsilon W))$ bits
for a count-based sliding window
even if we allow a relative error of at most $\varepsilon$
\cite{DatarGIM02, GibbonsT02}.  In fact, the whole data stream model
can be viewed as a special case of the sliding window model with
window size being infinite.  Also,
a count-based window is a special case of a time-based
  window in which exactly one item arrives at a time.
This paper focuses on time-based window, and
the algorithms are applicable to the other two models.

We study algorithms that enable the root to answer
three types of classical
$\varepsilon$-approximate queries, defined as follows.
Let $0 < \varepsilon < 1$. For any stream $\sigma$,
let $c_{j,\sigma}$
and $c_\sigma$ be the count of
item $j$ and all items
  {whose timestamps are} in the current window,
respectively.
Denote $c_j = \sum_\sigma c_{j,\sigma}$ and $c = \sum_\sigma c_\sigma$
as the total count of item $j$ and all items in all the data streams,
respectively.

\begin{itemize}
\item{\it Basic Counting.}
Return an estimate $\hat c$ on the total count $c$
such that $|\hat{c} - c| \le \varepsilon
c$.  (Note that this query can be generalized to count
data items of a fixed subset $X \subseteq U$;
the literature often refers to the special case with
$U=\{0,1\}$ and $X=\{1\}$.)
\item{\it Frequent Items.}
Given any $0 < \phi < 1$,
return
a set $F\subseteq U$ which includes all items $j$ with $c_j \ge \phi c$
and possibly some items $j'$ with $c_{j'} \ge \phi c - \varepsilon c$.

\item{\it Quantiles}.
Given any $0 < \phi < 1$, return an
item whose rank is in $\bigl[ \phi c - \varepsilon c,\, \phi c +
  \varepsilon c \bigr]$
among the $c$ items in the current sliding window.
\end{itemize}
  {As in most previous works,
we need to answer the
  following type of $\varepsilon$-approximate
queries in order to answer queries on frequent items.}
\begin{itemize}
\item{\it Approximate Counting.}
Given any item $j$,
return an estimate $\hat c_j$ such that $|\hat c_j - c_j| \le \varepsilon c$.
(Note that this query gives estimate for any item, not just the
frequent items.   {Also, the error bound is in term
of $c$, which may be much larger than $c_j$.})
\end{itemize}

We need an algorithm to determine when and how
the remote sites communicate with the root so that the
root can answer the queries at any time.
The objective is to minimize the worst-case communication cost
within a window of $W$ time units.

\vspace{.5ex}
{\bf Previous works.}
Recently, the database literature has a flurry of results on continuous
monitoring of distributed data streams, e.g.
\cite{OlstonJW03,GreenwaldK04,DasGGR04,SharfmanSK06,
ManjhiSDO05,CormodeGMR05,BabcockO03,JainYDZ05,MouratidisBP06,CormodeG05}.
The algorithms studied can be classified into two
types: {\em one-way} algorithms only allow messages
sent from each remote site to the root, and
{\em two-way} algorithms allow bi-directional
communication between the root and each site.
One-way algorithms are often very simple
as a remote site has little
information and all it can do is to update the root
when its local statistics deviate significantly from those previously sent.
On the other hand,
most two-way algorithms are complicated and
often involve non-trivial heuristics.
It is commonly believed in the database community
that two-way algorithms
are more efficient; however, for most
existing two-way algorithms, their worst-case communication costs
are still waiting for rigorous mathematical analysis, and
existing works often rely on experimental results when evaluating
the communication cost.


The literature contains several results on the
mathematical analysis of the worst-case
performance of one-way algorithms.
They are all for the whole data stream setting.
Keralapura et al.\ \cite{KeralapuraCR06}
studied the thresholded-count problem, which leads to
an algorithm for basic counting with communication cost
$O(\frac{k}{\varepsilon}\log \frac{N}{k})$ words, where $k$ and
$N$ are the number of streams and the
number of items in these streams, respectively.
Cormode~et~al.\ \cite{CormodeGMR05} gave an algorithm for
quantiles with communication cost $O(\frac{k}{\varepsilon^2} \log
\frac{N}{k})$ words per stream.
They also showed how to handle frequent items via
a reduction to quantiles,
so the communication cost remains the same.
More recently,  Yi and Zhang \cite{YiZ08} have reduced
the communication cost for frequent items
to $O(\frac{k}{\varepsilon} \log \frac{N}{k})$ words,
and quantile to $O(\frac{k}{\varepsilon} \log^2 (\frac{1}{\varepsilon})
\log \frac{N}{k})$ words, using some two-way algorithms; these are the
only analyses for two-way algorithms so far.

There have been attempts to devise heuristics
to extend some whole-data-stream algorithms to
sliding windows, yet not much has been known about
their worst-case performance.
For example, Cormode~et al.\ \cite{CormodeGMR05} have
extended their algorithms for quantiles and frequent items
to sliding windows.
They believed that the communication cost would only have
a mild increase,
but no supporting analysis has been given.  The
analysis of sliding-window algorithms is more difficult because
the expiry of items destroys some monotonic
property that is important to the analysis for
whole data stream.
In fact,
finding sliding-window algorithms with efficient worst-case
communication has been posed as an open problem
in the latest work of Yi and Zhang \cite{YiZ08}.

\vspace{0.5ex}
{\bf Our results.}
This paper gives the first mathematical analysis of
the communication cost in the sliding window model.
We derive
lower bounds on the worst-case communication cost
of any two-way
algorithm (and hence any one-way algorithm) for answering the four types
of $\varepsilon$-approximate queries.
These lower bounds hold even when each remote site
has unlimited space to maintain the local statistics exactly.
More interestingly, we analyze
some common-sense algorithms that use one-way communication only
and prove that their communication costs
match or nearly match the corresponding lower bounds.
In our algorithms, each remote site only needs to
maintain some $\Theta(\varepsilon)$-approximate statistics for its local data,
which actually adds more complication
to the problem.
These results demonstrate optimal or near optimal
communication-accuracy tradeoffs for supporting these queries
over the sliding window.
Our work reveals that
two-way algorithms could not
be much better than one-way algorithms in the worst case.

\renewcommand{\baselinestretch}{0.9}
\begin{table}\small
\center{
\begin{tabular}{| c || c | c | c |}
\hline
& {Basic Counting} \rule[-2.5mm]{0mm}{7mm}  & Approximate Counting/ & Quantiles  \\
& \raisebox{1ex}[0mm][0mm]{({\small\bf bits}) }  &
  \raisebox{1ex}[0mm][0mm]{Frequent items (\small\bf words)} &
  \raisebox{1ex}[0mm][0mm]{({\small\bf words})} \\
\hline \hline
\raisebox{-1ex}[0mm][0mm]{Whole data} &
 \raisebox{-0.2ex}[0mm][0mm]{$O( \frac{k}{\varepsilon} \log
   \frac{N}{k})$ words  \cite{KeralapuraCR06}}
 & $O(\frac{k}{\varepsilon}\log \frac{N}{k})$ \cite{YiZ08}
 & \rule[-2.5mm]{0mm}{8mm}
\raisebox{0.4ex}[0mm][0mm]{$O(\frac{k}{\varepsilon}
 \log^{2}(\frac{1}{\varepsilon})\log  \frac{N}{k})$ \cite{YiZ08}}

\\ \cline{3-4}
\raisebox{1ex}[0mm][0mm]{stream}&
  \raisebox{0.5ex}[0cm][0cm]{
 $\Theta( \frac{k}{\varepsilon} \log
 \frac{\varepsilon N}{k})$ bits
  }
&  \multicolumn{2}{c|}{\rule[-2.5mm]{0mm}{7mm}$\Omega(\frac{k}{\varepsilon}
\log \frac{\varepsilon N}{k} )$ \cite{YiZ08,YiZ09}} \\ \hline
\raisebox{-1.9ex}[0mm][0mm]{Sliding window}  &
\rule[-2.5mm]{0mm}{7mm}
\raisebox{-1.9ex}[0cm][0cm]{$\Theta(\frac{k}{\varepsilon} \log
  \frac{\varepsilon N}{k})$}
&
$O(\frac{k}{\varepsilon}\log \frac{N}{k})$ &
$O(\frac{k}{\varepsilon^2}\log
\frac{N}{k})$ \\ \cline{3-4}
\raisebox{0ex}[0mm][0mm]{}  & & \multicolumn{2}{c|}{
\rule[-2.5mm]{0mm}{7mm}
$\Omega(\frac{k}{\varepsilon}\log \frac{\varepsilon N}{k})$}
\\ \hline
\raisebox{-2ex}[0mm][0mm]{Sliding window} \rule[-4mm]{0mm}{9mm} &
$O( (\frac{W}{W-\tau})\frac{k}{\varepsilon} \log \frac{\varepsilon N}{k})$
&
$O((\frac{W}{W-\tau})\frac{k}{\varepsilon}\log \frac{N}{k})$ &
$O((\frac{W}{W-\tau})\frac{k}{\varepsilon^2}\log \frac{N}{k})$ \\
\cline{2-4}
\raisebox{1ex}[0mm][0mm]{\& out-of-order } &
$\Omega(\max\{\frac{W}{W-\tau},\frac{k}{\varepsilon}\log
    \frac{\varepsilon N}{k}\})$
&
\multicolumn{2}{c|}{\rule[-3.3mm]{0mm}{8.5mm}$\Omega(\max\{\frac{W}{W-\tau},
  \frac{k}{\varepsilon}\log \frac{\varepsilon N}{k}\})$}
\\ \hline
\end{tabular}
\caption{Bounds on the communication costs.
Note that the bounds are stated in bits for basic counting,
and in words for the other problems.}
\label{table:summary}
}
\end{table}
\renewcommand{\baselinestretch}{1}

Below we state the lower and upper bounds precisely.
  {Recall that there are $k$ remote sites and
the sliding window contains $W$ time units.
We prove that}
 {within any window,}
  {the root and the remote sites need to communicate,
in the worst case,
$\Omega(\frac{k}{\varepsilon}\log \frac{\varepsilon N}{k})$
bits  for basic counting
and $\Omega(\frac{k}{\varepsilon}\log \frac{\varepsilon N}{k})$
words for the other three queries, where $N$ is the total number of items
arriving or expiring within that
window.}\footnote{  {Note
that the number of items arriving or expiring within window
$[t-W+1,t]$ is no greater than the number of items arriving within
$[t-2W+1,t]$.}}
For upper bounds, our analysis shows that
basic counting requires $O(\frac{k}{\varepsilon}\log
\frac{\varepsilon N}{k})$ bits within any window,
and approximate counting $O(\frac{k}{\varepsilon} \log \frac{N}{k})$ words.
The estimates given by approximate counting are
sufficient to find frequent items, hence the latter problem
has the same communication cost.
For quantiles, it takes $O(\frac{k}{\varepsilon^2}\log \frac{N}{k})$ words.
See the second row (sliding window)
of Table~\ref{table:summary} for a summary.

As mentioned before, sliding-window algorithms can be
applied to handle the special case of whole data streams
  {in which the window size $W$ is infinite and $N$
  is the total number of arrived items.}
The first row of Table~\ref{table:summary} shows
the results on whole data streams.
Our work has improved the communication cost for
basic counting from $O( \frac{k}{\varepsilon}\log \frac{N}{k})$
words~\cite{KeralapuraCR06}
to $O(\frac{k}{\varepsilon}\log\frac{\varepsilon N}{k})$ bits.
For approximate counting and frequent items, our work implies
a one-way algorithm with communication cost of
$O(\frac{k}{\varepsilon} \log \frac{N}{k})$
words; this matches the performance of
the two-way algorithm by Yi and Zhang~\cite{YiZ08}.
In their algorithm, the root
regularly updates every remote site about the global count of all items.
In contrast,
we use the idea that items with small count could
be ``turned off'' for further updating.
As a remark, our upper bound on quantiles is
$O(\frac{k}{\varepsilon^2} \log \frac{N}{k})$ words
which is weaker than that of~\cite{YiZ08}.

Our algorithms can be readily applied
to out-of-order streams \cite{BuschT07,CormodeKT08}.
 {In an out-of-order stream, each item is associated
  with an integral time-stamp recording its creation time, which may
  be different from its arrival time.
We say that the stream has \emph{tardiness} $\tau$ if
any item with time-stamp $t$ must arrive within $\tau$ time
units from $t$, i.e., at any time in $[t, t+ \tau]$.
Without loss of generality, we assume that $\tau \in \{0, 1, 2,\dots,
W-1\}$ (if an item time-stamped at $t$ arrives after $t + W-1$,
it has already expired and can be ignored).
Note that for any data stream with tardiness greater than zero, the
items may not be arriving in non-decreasing order of their time-stamps.
Our previous discussion of
data streams assumes tardiness equal to $0$, and
such data streams are called \emph{in-order} data streams.}
The previous lower bounds for in-order streams are all valid
in the out-of-order setting. In addition,
we obtain lower bounds related to~$\tau$,
namely, $\Omega(\frac{W}{W-\tau})$ bits for basic counting and
$\Omega(\frac{W}{W-\tau})$ words for the other three problems.
Regarding upper bounds,
our algorithms when applied to
out-of-order streams with tardiness $\tau$ will just
increase the communication
cost by a factor of $\frac{W}{W-\tau}$.
The results are summarized in the last row of Table~\ref{table:summary}.

\vspace{0.5ex}
The idea for basic counting is relatively simple.
As the root does not require an exact total count, each data
stream can communicate to the root only when its local count
increases or decreases by a certain ratio $\varepsilon > 0$; we call such a
communication step an \emph{up} or \emph{down} event, respectively.
To answer the total count of all streams, the root simply sums up
all the individual counts it has received. It is easy to prove that
this answer is within some desired error bound.
If each count is over the whole stream (i.e., window size = $\infty$
and   {$N$ is the total number of arrived items}),
the count is increasing and there is no down event.
A stream would have at most $O(\log_{1+\varepsilon}N)$ up events
and the communication cost is at most that many words.
However, the analysis becomes non-trivial in a sliding time window.
Now items can expire and down events can occur.
An up event may be followed by some
down events and the count is no longer increasing.
The tricky part is to find a new measure of progress.
We identify a ``characteristic set''
of each up event such that
each up event must increase the size of
this set by a factor of at least $1+\varepsilon$,
hence bounding the number of up events to be $O(\log_{1+\varepsilon} N)$.
Down events are bounded using another characteristic set.
Due to space limitation, the details can only be given in the full paper.

Approximate counting of all possible items is much more complicated,
which will be covered in details in the rest of this paper.
Assuming in-order streams, we derive and analyze two algorithms
 for approximate counting in Section~\ref{sec:fi}.
In Section~\ref{sec:extension}, we discuss frequent items, quantiles, and
finally out-of-order streams.
The lower bound results 
are relatively simple and omitted due to space limitation.

\section{Approximate Counting of all items}
\label{sec:fi}
This section presents algorithms for the streams to communicate to the root
so that the root at any time can approximate
the count of each item.
As a warm-up, we first consider the simple algorithm
in which a stream will inform the root
whenever its count of an item increases or decreases
by a certain fraction of its total item count.
  {We show
in Section~\ref{sect:simpleAlg} that within any window of $W$ time units,
each data stream $\sigma_i$ ($1 \leq i \leq k$) needs to send at most
$O((\Delta + \frac{1}{\varepsilon}) \log n_i)$ words to the
root, where $\Delta$ is the number of distinct items and $n_i$
is the number of items of $\sigma_i$ that arrive or expire within the
window.
Then, the total communication cost within this window is $\sum_{1 \leq
  i \leq k} (\Delta + \frac{1}{\varepsilon})\log n_i$, which, by
Jensen's inequality, is no greater than $(\Delta +
\frac{1}{\varepsilon})k \log (\sum_{1 \leq i \leq k} n_i)/k =
(\Delta + \frac{1}{\varepsilon})k \log \frac{N}{k}$
where $N = \sum_{1 \leq i \leq k}n_i$.}
We then modify the algorithm so that
a stream can ``turn off'' items whose counts are too small,
and we give a more complicated analysis
to deal with the case when many such items increase their counts rapidly
(Section~\ref{sect:fullAlg}).
The communication cost is reduced to $O(\frac{k}{\varepsilon} \log
\frac{N}{k})$ words,  {independent of $\Delta$.}

\subsection{A simple algorithm}
\label{sect:simpleAlg}
Consider any stream $\sigma$.  At any time $t$,
let $\ns(t)$ and $\cjs(t)$ be the number of all items and item $j$
arriving at $\sigma$ in $[t-W+1, t]$, respectively.
Let $\lambda < 1/11$ be a positive constant
(which will be set to $\varepsilon/11$).
We maintain two $\lambda$-approximate data structures \cite{DatarGIM02,LeeT06}
at $\sigma$ locally,
which can report estimates $\hns(t)$ and $\hcjs(t)$ for
$\ns(t)$ and $\cjs(t)$, respectively,
such that
\footnote{
The constant 6 in the inequality is arbitrary.
It can be replaced with any number provided that
other constants in the algorithm and analysis
(e.g., the constant 9 in definition of up events)
are adjusted accordingly.
}
$$
(1 - \lambda / 6)\ns(t) \le \hns(t) \le (1+\lambda / 6) \ns(t);
\mbox{ \hspace{.1in} and \hspace{.1in} }
\cjs(t) - \lambda \ns(t) \le \hcjs(t) \le \cjs(t) + \lambda \ns(t).
$$

\vspace{.5ex}
\begin{center}
\begin{minipage}{.9\textwidth}
\hrulefill

{\bf Simple algorithm.}
At any time $t$, for any item $j$,
let $p < t$ be the last time $\hcjs(p)$ is sent
to the root. The stream sends the estimate $\langle j,
\hcjs(t)\rangle$
to the root if the following event occurs.
\begin{itemize}
\item
{\it Up:}\hspace*{3.1ex}  $\hcjs(t) > \hcjs(p) + 9\lambda \hns(t)$.

\item
{\it Down:}  $\hcjs(t) < \hcjs(p) - 9\lambda \hns(t)$.
\end{itemize}
\vspace*{-.1in}
\hrulefill
\end{minipage}
\end{center}

\vspace{1ex}
\noindent{\bf Root's perspective.}
At any time $t$,
let $\rjs(t)$ be the last estimate received from a stream
$\sigma$ for item $j$ (at or before $t$).  The root can estimate
the total count of item $j$ over all streams by
summing all $\rjs(t)$ received.
More precisely, for any $0 < \varepsilon < 1$,
we set $\lambda = \varepsilon / 11$ and let each stream
use the simple algorithm.
Then for each stream $\sigma$,
the approximate data structures for $\hat c_j(t)$ and
$\hns(t)$ together with the simple algorithm guarantee that
$
  \cjs(t) -11 \lambda \ns(t) \le \rjs(t) \le \cjs(t) + 11 \lambda \ns(t).
$
Summing $\rjs(t)$ over all streams
would give the root an estimate of the total count of item $j$
within an error of $\varepsilon$ of the total count of all items.

\vspace{.5ex}
\noindent{\bf Communication Complexity.}
At any time $t$, we denote the reference window as
$[t_o, t]$, where $t_o = t-W + 1$.
 {Let $n$ be the number of items of $\sigma$ that
arrive or expire in $[t_o, t]$.}
Assume that there are at most $\Delta$ distinct items.
We first show that a stream $\sigma$ encounters
$O((\frac{1}{\lambda} + \Delta) \log n)$ up events
and sends $O((\frac{1}{\lambda} + \Delta) \log n)$ words
within $[t_o,t]$.  The analysis of down events is similar and will
be detailed later. For any time $t_1 \le t_2$,
it is useful to define $\sigma_{[t_1,t_2]}$
(resp.\ $\sigma_{j,[t_1,t_2]}$)
as the multi-set of all items (resp.\ item $j$ only)
arriving at $\sigma$ within $[t_1,t_2]$, and $|\sigma_{[t_1,t_2]}|$
as the size of this multi-set.

Consider an up event $U_j$ of some item $j$
that occurs at time $v \in [t_o, t]$.
Define the {\em previous event} of $U_j$ to be
the latest event (up or down) of item $j$
that occurs at time $p < v$.
We call $p$ the \emph{previous-event time} of $U_j$.
The number of up events
with previous-event time before $t_o$
is at most~$\Delta$.
To upper bound the number of
up events with previous-event time $p \ge t_o$
is, however, non-trivial;
below we call such an up event a \emph{follow-up} (event).
Intuitively, a follow-up can be triggered by
frequent arrivals of an item, or mainly
the relative decrease of the total count.
This motivates us to classify follow-ups
into two types and analyze them differently.
A follow-up $U_j$ is said to be {\em absolute} if
$\ns(p) \le \frac{6}{5} \ns(v)$, and {\em relative} otherwise.
Define $\ri(U_j)$ to be
the multi-set of item $j$'s that arrive after the
previous event of $U_j$, i.e., $\ri(U_j) = \sigma_{j,[p+1, v]}$.

\vspace{.5ex}
{\bf Absolute follow-ups.}
To obtain a tight bound of absolute follow-ups, we need
a characteristic-set argument that can
consider the growth of different items together.
Let $t_1, t_2, ..., t_k$ be the times in $[t_o, t]$
when some absolute follow-ups (of one or more items) occur.
Let $x_i$ be the number of items having an absolute follow-up at $t_i$.
Note that for all $i$, $x_i \le \min\{1 / (7 \lambda), \Delta\}$,%
\footnote{
  If an up event of an item $j$ occurs at time $t_i$,
  then $\cjs(t_i) \ge \hcjs(t_i) - \lambda \ns(t_i)
  > 9 \lambda \hns(t_i) - \lambda \ns(t_i) \ge 7 \lambda \ns(t_i)$.
  Thus the number of up events at time $t_i$ is at most
  $\ns(t_i) / (7 \lambda \ns(t_i)) = 1/ (7\lambda)$.
}  
and $\sum_{i=1}^{k} x_i$ is the number of absolute follow-ups in $[t_o, t]$.
We define the characteristic set $S_i$ at each $t_i$ as follows:
$$\mbox{$S_i$ = the union of
$\ri(U_j)$ over all absolute follow-ups $U_j$ occurring at $t_1, t_2, \dots, t_i$.
}
$$



 {Recall that $n$ is the number of items of $\sigma$
  that arrive or expire in $[t-W+1, t]$.}

\begin{lemma} \label{lem:fi-step-up-s-simple}
{\bf (i)} For any $2 \le i \le k$, $|S_i| > (1 + 6 x_i \lambda)
|S_{i-1}|$.
{\bf (ii)} There are $\sum_{i=1}^{k} x_i = O(\frac{1}{\lambda} \log n)$
absolute follow-ups within $[t_o, t]$.
\end{lemma}

\begin{proof}
For (i), consider an absolute follow-up $U_j$ of an item $j$,
occurring at time $t_i$ with previous-event time $p_i$.
Note that the increase in the count of item $j$ from $p_i$ to $t_i$ must be due to the recent
items. We have
\begin{eqnarray*}
|\ri(U_j)|
& \ge & \cjs(t_i) - \cjs(p_i) \\
&\ge & \hcjs(t_i) - \hcjs(p_i) - \lambda \ns(t_i) - \lambda \ns(p_i)
\hspace{.545in}\mbox{(by $\sigma$'s local data structures)}\\
& > & 9 \lambda \hns(t_i) - \lambda \ns(t_i) - \lambda \ns(p_i)
\hspace{.995in}\mbox{(definition of an up event)}\\
& \ge & \textstyle(9 \lambda (1 - \frac{\lambda}{6}) - \lambda - \frac{6}{5} \lambda)\ns(t_i)
\ge 6 \lambda \ns(t_i) \hspace{.25in}\mbox{($U_j$ is absolute)}
\end{eqnarray*}
 There are $x_i$ absolute follow-ups at $t_i$,
so $|S_{i}| > |S_{i-1}| + x_i \left(6 \lambda \ns(t_i)\right)$.
Since $S_i \subseteq \sigma_{[t_o, t_i]}$,
$\ns(t_i) \ge |S_i| \ge |S_{i-1}|$.
Therefore, we have $|S_i| > |S_{i-1}| + 6 x_i \lambda |S_i|
\ge (1 + 6 x_i \lambda) |S_{i-1}|$.

For (ii), we note that $n \ge |S_k|
> \prod_{i = 2}^{k} (1 + 6  x_i \lambda) |S_{1}|$,
and $|S_1| \ge 1$.
Thus, $\prod_{i = 2}^{k} (1 + 6 x_i \lambda) < n$,
or equivalently, $\ln n > \sum_{i = 2}^{k} \ln(1 + 6 x_i \lambda)$.
The latter is at least
$\sum_{i=2}^{k} \frac{6 x_i\lambda}{1+ 6 x_i\lambda}
\geq
\lambda \sum_{i=2}^{k} x_i$.
The last inequality follows from that $x_i \le 1 / (7 \lambda)$ for all $i$.
Thus, $\sum_{i=1}^{k} x_i
\le x_1 + \frac{1}{\lambda} \ln n
= O( \frac{1}{\lambda} \log n)$.
\end{proof}

{\bf Relative follow-ups.}
A relative follow-up occurs only when a lot of items expire,
and relative follow-ups of the same item cannot occur
too frequently. Below we define $O(\log n)$ time intervals
and argue that no item can have two relative follow-ups within
an interval.
For an item with time-stamp $t_1$,
we define the \emph{first expiry time} to be $t_1+W$.
At any time $u$ in $[t_o,t]$, define $H_u$
to be the set of all items whose
first expiry time is within $[u+1, t]$,
i.e., $H_u = \sigma_{[u-W+1,t_o-1]}$.
$|H_u|$ is non-increasing as $u$ increases.
Consider the times $t_o = u_0 < u_1 < u_2 < \dots< u_\ell \le t$
such that for $i \ge 1$, $u_i$ is the first time
such that $|H_{u_i}| < \frac{5}{6}|H_{u_{i-1}}|$.
For convenience, let $u_{\ell + 1} = t+1$.
Note that $| H_{u_0} | \le n$ and $\ell = O(\log n)$.

\begin{lemma}\label{lem:fi-type2-up}
{\bf (i)} Every item $j$ has at most one relative follow-up $U_j$
within each interval $[u_i, u_{i+1}-1]$.
{\bf (ii)} There are at most $O( \Delta \log n)$
relative follow-ups within $[t_o, t]$.
\end{lemma}
\begin{proof}
For (i), assume $U_j$ occurs at time $v$ in $[u_i, u_{i+1}-1]$, and its previous event
occurs at time $p$.
By definition, $\ns(p) > \frac{6}{5}\ns(v)$. Thus,
$$\textstyle
|H_p| - |H_v|
= | \sigma_{[p-W+1, v-W]} |
\ge \ns(p) - \ns(v) > \frac{1}{5} \ns(v)
\ge \frac{1}{5}| \sigma_{[v-W+1, t_o-1]}| = \frac{1}{5}|H_v| \enspace,
$$
and $|H_v| < \frac{5}{6} |H_p|$.
Since  $v < u_{i+1}$ and $|H_v| \ge \frac{5}{6}|H_{u_i}|$,
we have $|H_p| > |H_{u_i}|$ and $p < u_i$.
For (ii), there are $\Delta$ distinct items, so there
are at most $\Delta$ relative follow-ups
within each interval $[u_i, u_{i+1}-1]$,
and  at most $O(\Delta \log n)$ relative follow-ups within $[t_o, t]$.
\end{proof}

{\bf Down events.} The analysis is symmetric to that
of up events. The only non-trivial thing is the definition of the
characteristic set for
bounding the absolute follow-downs $D_j$, which
is defined in an opposite sense:
Assume $D_j$ occurs at time $v$ and its  previous event occurs at $p \ge t_o$.
$D_j$ is said to be {\em absolute}
if $\ns(p) \le \frac{6}{5} \ns(v)$.
Let $\ei(D_j)$ be
the multi-set of item $j$'s whose first expiry time is within
$[p+1,v]$. I.e., $\ei(D_j) = \sigma_{j,[p-W+1, v-W]}$.

It is perhaps a bit tricky that
instead of defining the characteristic set of absolute follow-downs
at the time they occur,
we consider the times of the corresponding
\emph{previous events} of these follow-downs.
Let $p_1, p_2, ..., p_{k}$ be the times in $[t_o, t]$ such that
there is at least one event $E_j$ (up or down) at $p_i$ which is the
previous event of an absolute follow-down $D_j$ occurring after $p_i$.
Let $y_i$ be the number of such previous events at $p_i$, and
let $AD(p_i)$ be the set of corresponding absolute follow-downs.
Note that $y_i$ (unlike $x_i$) only admits
a trivial upper bound of $\Delta$.
We define
the characteristic set $T_i$ for each $p_i$ as follows:
$$\mbox{$T_i$ = the union of $\ei(D_j)$ over all
$D_j \in AD(p_i), AD(p_{i+1}), \dots, AD(p_k)$.
}$$
Similar to Lemma~\ref{lem:fi-step-up-s-simple}, we can show
that $|T_{i}| > (1 + 5 y_i \lambda) |T_{i+1}|$. 
Owing to a weaker bound of individual $y_i$, the number
of absolute follow-downs, which equals
$\sum_{i=1}^{k} y_i$, is shown to be
$O((\frac{1}{\lambda} + \Delta) \log n)$.

\vspace{.5ex}
Combining the analyses on up and down events,
and let $\lambda = \varepsilon / 11$,  we have the following.
\begin{theorem}\label{thm:simple}
  {
The simple algorithm
sends at most $O( (\frac{1}{\varepsilon} + \Delta) \log n)$ words
to the root during window $[t-W+1, t]$.}
\end{theorem}

\vspace*{-1ex}
\subsection{The full algorithm}
\label{sect:fullAlg}
\vspace*{-1ex}
In this section, we extend the previous algorithm and
give a new characteristic-set analysis that is based on
future events (instead of the past events) to show that
  {each stream's communication cost per window can be
  reduced to
$O(\frac{1}{\varepsilon} \log n)$ words.  Then,
by Jensen's inequality again, we conclude that the total
communication cost
per window is $O(\frac{k}{\varepsilon} \log \frac{N}{k})$.}
Intuitively, when the estimate $\hcjs(t)$ of an item $j$ is too small,
say, less than $3 \lambda \hns(t)$,
the algorithm treats this estimate as 0 and
set the  $\off$ flag of $j$ to be true.
This restricts the number of items with a positive estimate
to $O(\frac{1}{\lambda})$.
Initially, the $\off$ flag is true for all items $j$.
Given $0 < \lambda < $   {$\varepsilon /11$},
 the stream communicates with the root as follows.

\vspace{.5ex}
\begin{center}
\begin{minipage}{.9\textwidth}
\hrulefill

{\bf Algorithm AC.}
At any time $t$, for any item $j$,  {let
  {$p < t$} be the time the
last estimate of $j$, i.e., $\hcjs(p)$, is
sent to the root.}
The stream sends the estimate of $j$ to the root
if the following event occurs.
 \begin{itemize}
 \item
 {\it Up:} \hspace{3ex}If $\hcjs(t) >  {\hcjs(p)} + 9\lambda \hns(t)$,
 send $\langle j, \hcjs(t) \rangle$ and set $\off = \false$~.
 \item
 {\it Off:} \hspace{2.7ex}If
  $\off = \false$ and $\hcjs(t) < 3 \lambda \hns(t)$,
  {reset $\hcjs(t)$ to 0,
  send $\langle j, \hcjs(t)\rangle$}\\
  \mbox{}\hspace{.5in}{and set $\off=\true$.}
 \item
 {\it Down:} If $\off = \false$ and $\hcjs(t) <
  {\hcjs(p)} - 9\lambda \hns(t)$,
 send $\langle j,  \hcjs(t) \rangle$.
 \end{itemize}
 \vspace*{-.1in}
\hrulefill
 \end{minipage}
\end{center}
\vspace*{.05in}
\vspace{.5ex}

It is straightforward to check that the root can
answer the approximate counting query for any item.
We analyze the communication complexity of different events as follows.

\begin{fact} \label{lem:rj}
At any time $v$, the number of items $j$ with $\off = \false$ is at most
$\frac{1}{\lambda}$.\footnote{%
  For any item $j$, if $\off = \false$,
  then $\hcjs(v) \ge 3 \lambda \hns(v)$ and
  $\cjs(v) \ge \hcjs(v) - \lambda \ns(v)
  \ge (3 \lambda (1  - \lambda) - \lambda) \ns(v)
  \ge \lambda \ns(v)$.  Thus
  the number of items $j$ with $\off = \false$ is at most
  $\ns(v) / \lambda \ns(v) = \frac{1}{\lambda}$.
} 
\end{fact}

\noindent{\bf Off events.}
  {
Recall that we are considering the window $[t_o,t]$, and
$n$ is the number of items arriving or expiring within $[t_o,t]$.}
By Fact~\ref{lem:rj},
just before $t_o$,
there are at most $\frac{1}{\lambda}$ items with $\off = \false$.
Within $[t_o, t]$, only an up event can set the {\em off}\/ flag
to false.
Thus the number of off events within  $[t_o, t]$
is bounded by $\frac{1}{\lambda}$ plus the number of up events.


\vspace{.5ex}
\noindent{\bf Up and Down events.}
The assumption of $\Delta$ gives a trivial bound on
 those events involving  items with very small counts and in
particular, those up events immediately following the off events.
Such up events are called {\em poor-up} events or simply {\em poor-ups}.
Using the {\em off}\/ flag, we can easily adapt
the analysis of the simple algorithm to bound
all the down and up events of the full algorithm,
but except the poor-ups. The following simple
observations, derived from Fact 1, allow us to replace
$\Delta$ with $1/\lambda$ in the previous analysis
to obtain a tighter upper bound of
$O(\frac{1}{\lambda} \log n)$.  Let $v$ be any time in $[t_o,t]$.

\vspace{.5ex}
\begin{itemize}
\item
There are at most $1/\lambda$ items
whose first event after $v$ is a down event.
\vspace{.5ex}
\item
There are at most $1/\lambda$ non-poor-up
events after $v$ whose previous event is before $v$.
\end{itemize}
\vspace{.5ex}

It remains to analyze the poor-ups.
Consider a poor-up $U_j$ at time $v$ in $[t_o, t]$.
By definition, $\off = \false$ at time $v$.
The trick of analyzing $U_j$'s
is to consider when the corresponding
items will be ``off'' again instead
of what items constitute the up events.
Then a characteristic
set argument can be formulated easily.
Specifically, we first observe that,
by Fact~1, there are at most $\frac{1}{\lambda}$ poor-ups
whose {\em off}\/ flags remain false up to time $t$.
Then it remains to consider those $U_j$
whose {\em off}\/ flags will be set to true at some time $d \le t$.
Below we refer to $d$ as the \emph{first off time} of $U_j$.


\vspace{.5ex}
\noindent{\bf Poor-up with early off.}
Consider a poor-up $U_j$ that occurs at time $v$ in $[t_o, t]$
and has its first off time at $d$ in $[v+1,t]$.
Let $\eif(U_j)$ be all the item $j$ whose first expiry time is
within $[v+1, d]$.  I.e.,  $\eif(U_j) = \sigma_{j,[v+1 - W, d -W]}$.
As an early off can be due to the expiry of many copies of item $j$
or the arrival of a lot of items, it is natural to divide
the poor-ups into two types:  with an \emph{absolute} off if
$\ns(d) \le \frac{6}{5} \ns(v)$, and \emph{relative} off otherwise.
For the case with absolute off, we consider
the distinct times $t_1, t_2, \dots, t_k$ in $[t_o, t$]
when such poor-ups occur.
Let $x_i$ be the number of such poor-ups at time $t_i$.
Note that $x_i \le 1 / (7 \lambda)$.
For each time $t_i$,
we define the characteristic set
$$\mbox{$F_i$ =  the union of
$\eif(U_j)$ over all $U_j$ occurring at $t_i, t_{i+1}, \dots, t_k$.
}
$$

\begin{lemma} \label{lem:fi-poor-up-s}
{\bf (i)}
For any $1 \le i \le k-1$, $|F_{i}| > (1 + x_i \lambda) |F_{i+1}|$.
{\bf (ii)}
Within $[t_o, t]$, there are $\sum_{i=1}^{k} x_i = O(
\frac{1}{\lambda} \log n )$
poor-ups each with an absolute off.
\end{lemma}

\begin{proof}
For (i),
consider an item $j$ and a poor-up $U_j$ with an absolute off that occurs at time $t_i$
and has its first off at time $d_i$.
The decrease in $\cjs$ 
must be due to expiry of item~$j$.
\begin{eqnarray*}
|\eif(U_j)|
& \ge & \cjs(t_i) - \cjs(d_i)
\ge \hcjs(t_i) - \hcjs(d_i) - \lambda \ns(t_i) - \lambda \ns(d_i) \\
& > & 9 \lambda \hns(t_i) - 3 \lambda \hns(d_i) - \lambda \ns(t_i) - \lambda \ns(d_i)
\hspace{.52in}\mbox{(definition of up and off)}\\
& \ge & \textstyle(9 \lambda (1 - \frac{\lambda}{6}) - \lambda)\ns(t_i)
- (3 \lambda(1 + \frac{\lambda}{6}) + \lambda) \ns(d_i)
\ge 7 \lambda \ns(t_i) - 5 \lambda \ns(d_i) \\
& \ge & \textstyle (7  - 5 (\frac{6}{5}) )\lambda \ns(t_i) = \lambda \ns(t_i)
\hspace{1.15in}\mbox{(definition of absolute off)}
\end{eqnarray*}
Thus,
$|F_{i}| > |F_{i+1}| + x_i \left(\lambda \ns(t_i)\right)$.
Since $F_i \subseteq \sigma_{[t_i - W + 1, t - W]}$,
$|F_i| \le \ns(t_i)$. Therefore,
$|F_i|
> |F_{i+1}| +  x_i \lambda |F_i|
> (1 + x_i \lambda) |F_{i+1}|$.
By (i), we can prove (ii) similarly to Lemma~\ref{lem:fi-step-up-s-simple} (ii).
\end{proof}

Analyzing poor-ups with a relative off is again based
on an isolating argument.
We divide $[t_o,t]$ into $O(\log n)$ intervals
according to how fast the total item count starting from $t_o$ grow;
specifically, we want two consecutive time boundaries $u_{i-1}$ and $u_i$
to satisfy $|\sigma_{[t_o, u_i]} | > \frac{6}{5} | \sigma_{[t_o, u_{i-1}]} |$.
Then we  show that for any poor-up  within $[u_{i-1},u_i - 1]$,
its relative off, if exists, occurs
at or after $u_i$.  Thus
there are at most $\frac{1}{\lambda}$ such poor-ups within each interval
and a total of $O(\frac{1}{\lambda}  \log n)$ within $[t_o,t]$.

\begin{lemma}\label{lem:fi-type2}
{\bf (i)} Consider a poor-up $U_j$ with a relative off.  Suppose it occurs at time $v$
in $[t_o,t]$, and its first off time is at $d$ in $[v+1,t]$.
Then $|\sigma_{[t_o, d]}| > \frac{6}{5}| \sigma_{[t_o, v]} |$.
{\bf (ii)} Within $[t_o,t]$, there are at most $O( \frac{1}{\lambda}
\log n)$ poor-ups each with a relative off.
\end{lemma}

\begin{proof}
For (i), by the definition of a relative off,
$\ns(d) > \frac{6}{5}\ns(v)$.
Thus,
$|\sigma_{[t_o, d]} | - | \sigma_{[t_o, v]} |
= | \sigma_{[v+1, d]} |
\ge \ns(d) - \ns(v) > \frac{1}{6} \ns(d)
\ge \frac{1}{6}| \sigma_{[t_o, d]}|$.
This implies
$|\sigma_{[t_o, d]}| > \frac{6}{5}| \sigma_{[t_o, v]} |$.

For (ii), consider the times $t_o = u_0 < u_1 < u_2 < \dots< u_\ell \le t$
such that for $i \ge 1$,
$u_i$ is the first time such that
$|\sigma_{[t_o, u_i]} | > \frac{6}{5} | \sigma_{[t_o, u_{i-1}]} |$.
For convenience, let $u_{\ell + 1} = t + 1$.
Note that $| \sigma_{[t_o, t]} | \le n$ and $\ell = O(\log n)$.
Furthermore, for any time $v \in [u_{i-1}, u_i - 1]$,
$| \sigma_{[t_o, v]} | \le \frac{6}{5} |\sigma_{[t_o, u_{i-1}]} |$.
Therefore, by (i),
for any poor-up of an item $j$ within $[u_{i-1}, u_i - 1]$,
its relative off, if exists, occurs at or after $u_i$,
which implies
at time $u_i - 1$, $\cjs(u_i-1) \ge \lambda \ns (u_i - 1)$.
Then within each interval $[u_{i-1}, u_i - 1]$,
the number of such $j$ as well as the number of poor-ups with a relative off
are at most $\frac{1}{\lambda}$.
Within $[t_o, t]$, there are $\ell = O( \log n)$ intervals
and hence $O( \frac{1}{\lambda} \log n)$
poor-ups each with a relative off.
\end{proof}




\begin{theorem} \label{thm:fi-comm}
For approximate counting,
each individual stream can use the algorithm AC
with $\lambda = \varepsilon / 11$ and
it sends at most $O( \frac{1}{\varepsilon} \log n )$ words to the root
within a window.
\end{theorem}

{\bf Memory usage of each remote site.} 
Recall that we use two $\lambda$-approximate data structures
\cite{DatarGIM02,LeeT06}
for the total item count and individual item counts,
which respectively require
$O(\frac{1}{\lambda} \log^2(\lambda n))$ bits
and $O(\frac{1}{\lambda})$ words.
Note that $O(\frac{1}{\lambda} \log^2(\lambda n))$ bits
is equivalent to $O(\frac{1}{\lambda} \log (\lambda n))$ words.
Furthermore, at any time,
we only need to keep track of
the last estimate sent to the root of
all item $j$ with $\off = \false$,
which by Fact~\ref{lem:rj}, requires $O(\frac{1}{\lambda})$ words.
By setting $\lambda = \varepsilon / 11$ (see Theorem~\ref{thm:fi-comm}),
the total memory usage of a remote site
is $O(\frac{1}{\lambda} \log (\lambda n))
= O(\frac{1}{\varepsilon} \log (\varepsilon n))$ words.

\section{Extensions}
\label{sec:extension}
We extend the previous 
techniques to solve the problems of frequent items
and quantiles and  handle out-of-order streams.
Below BC refers to our algorithm for basic counting.

\vspace{.5ex}
{\bf Frequent items.}
Using the algorithms BC and AC, the root can
answer the $\varepsilon$-approximate frequent items as follows.
Each stream $\sigma$ communicates with the root using BC
with error parameter $\varepsilon/ 24$
and AC with error parameter $11 \varepsilon/ 24$.
At any time $t$,
let $\rs(t)$ and $\rjs(t)$ be the latest estimates of
the numbers of all items and item $j$, respectively,
received by the root from $\sigma$.  
To answer
a query of frequent items with threshold
$\phi \in (0, 1]$ at time $t$,
the root can return all items $j$ with $\sum_\sigma \rjs(t) \ge
(\phi - \frac{\varepsilon}{2}) \sum_\sigma \rs(t)$ as the set of frequent
items.

To see the correctness, let $\nsm(t)$ and $\cjsm(t)$
be the number of all items and item $j$ in $\sigma$ at time $t$, respectively.
Algorithm BC guarantees $| \rs(t) - \nsm(t) | \le \frac{\varepsilon}{24} \nsm(t)$,
and algorithm AC guarantees $| \rjs(t) - \cjsm(t) | \le \frac{11 \varepsilon}{24} \nsm(t)$.
Therefore, if an item $j$ is returned by the root, then
$ \sum_\sigma \cjsm(t)
\ge \textstyle \sum_\sigma \rjs(t) -
      \frac{11 \varepsilon}{24} \sum_\sigma \nsm(t)
\ge \textstyle (\phi - \frac{\varepsilon}{2}) \sum_\sigma \rs(t) -
      \frac{11 \varepsilon}{24}  \sum_\sigma \nsm(t)
\ge \textstyle (\phi - \frac{\varepsilon}{2}) (1 - \frac{\varepsilon}{24} ) \sum_\sigma \nsm(t)
      - \frac{11 \varepsilon}{24} \sum_\sigma \nsm(t)
\ge (\phi - \frac{\varepsilon}{2}  - \phi \frac{\varepsilon}{24}
- \frac{11 \varepsilon}{24}) \sum_\sigma \nsm(t)$
where 
the second inequality comes from the definition of the algorithm.
The last term above is at least $(\phi - \varepsilon) \sum_\sigma \nsm(t)$,
so $j$ is a frequent item.
If an item $j$ is not returned by the root, then $\sum_\sigma
\rjs(t) < (\phi - \frac{\varepsilon}{2}) \sum_\sigma \rs(t)$ and we can
show similarly that $\sum_\sigma \cjsm(t) < \phi \sum_\sigma \nsm(t)$.
%

\vspace{.5ex}
{\bf Quantiles.}
We give an algorithm for
$\varepsilon$-approximate quantiles queries.
Let $\lambda = \varepsilon /20$.
For each stream, we keep track of the $\lambda$-approximate $\phi$-quantiles
for $\phi = 5 \lambda, 10 \lambda, 15 \lambda, \dots, 1$.
We update the root for all these $\phi$-quantiles when one of the following two events occurs:
(i) for any $k$, the value of the $(5k\lambda)$-quantile is
larger than the value of the $(5(k+1) \lambda)$-quantile last reported to the root, or
(ii) for any $k$, the value of the $(5k\lambda)$-quantile is
smaller than the value of the $(5(k-1) \lambda)$-quantile last reported to the root.
The stream also communicates with the root using BC with
error parameter $\lambda$. In the root's perspective,
at any query time $t$, let $\phi \in (0,1]$ be the
query given and let $\rs(t)$ be the last estimate sent by
$\sigma$ for the number of all items.
The root sorts the quantiles last reported by all streams
and for each stream $\sigma$, gives a weight of $5 \lambda \rs(t)$ to each quantile of $\sigma$.
Then the root returns the smallest item $j$ in the sorted sequence such that
the sum of weights for all items no greater than $j$ is at least $\ceil{\phi \sum_\sigma \rs(t)}$.
Careful counting can show that $j$ is
an $\varepsilon$-approximate $\phi$-quantile.
  {To bound the communication cost, let $n$ be
the number of items of $\sigma$ arriving or expiring during the window
$[t-W+1,t]$.}
We observe that
when an event occurs, many items have either arrived or expired
after the previous event.
Using similar analysis as before,
we can show that within a window, there are at most $O(\frac{1}{\varepsilon} \log n)$ such events
and thus each stream sends $O(\frac{1}{\varepsilon^2} \log n)$
words.
  {By Jensen's inequality again, our algorithm's
 total communication cost per window is $O(\frac{k}{\varepsilon^2}\log
 \frac{N}{k})$
where $N$ is the number of items of the $k$ streams
that arrive or expire within the window.}
Note that the lower bound of $O(\frac{1}{\varepsilon} \log
(\varepsilon n))$ words
for approximate frequent items carries to approximate quantiles,
as we can answer approximate frequent items using approximate
quantiles as follows. The root poses $\varepsilon$-approximate
$\phi$-quantile queries
for $\phi = \varepsilon, 2\varepsilon, \dots, 1$.
Given the threshold $\phi'$ for frequent items,
the root returns all items that repeatedly occur as $\frac{\phi'}{\varepsilon} - 2$ (or more) consecutive quantiles,
and these items are $(4\varepsilon)$-approximate frequent items.

\vspace{.5ex}
{\bf Out-of-order streams.}
All our algorithms
can be extended to out-of-order stream with a communication cost increased
by a factor of $\frac{W}{W-\tau}$, as follows.
Each stream uses the data structures for out-of-order streams
(e.g., \cite{BuschT07,CormodeKT08}) to maintain the local estimates.
Then each stream uses our communication algorithms for in-order streams.
It is obvious the root
can answer the corresponding queries.
For the communication cost,
consider any time interval $P = [t - (W - \tau) + 1, t]$ of size $W-\tau$.
Items arriving in $P$ must have time-stamps in $[t-W+1, t]$.
Using the same arguments as before, we can show
the same communication cost of each algorithm,
but only for a window of size $W-\tau$ instead of $W$.
Equivalently, in any window of size $W$,
the communication cost is increased by a factor of $O(\frac{W}{W-\tau})$.


\bibliographystyle{plain}

\end{document}